\documentclass[11pt]{article}

\usepackage{amssymb}
\usepackage{amsmath}
\usepackage[left]{lineno}
\usepackage{blindtext}
\usepackage{amstext, amsmath,latexsym,amsbsy,amssymb}
\usepackage{enumitem}

\usepackage[pdftex]{graphicx}
\usepackage{enumerate}
\usepackage{url}
\usepackage{enumitem}
\newtheorem{definition}{Definition}[section]

\newenvironment{proof}{{\bf Proof.\ }}{\QED\\}
\newtheorem{lemma}{Lemma}[section]

\numberwithin{equation}{section}
\newtheorem{theorem}{Theorem}[section]

\newtheorem{cnjctr}{Conjecture}[section]
\newcommand{\QED}{\hspace*{\fill}\rule{2.5mm}{2.5mm}}

\usepackage{amssymb}\usepackage{graphicx}
\usepackage{tikz}
\usepackage{color}
\newcommand\qed{\hfill$\sqcap\kern-7.5pt\hbox{$\sqcup$}$}

\newcommand{\beqn}{\begin{equation}}
\newcommand{\eeqn}{\end{equation}}
\newcommand{\bear}{\begin{eqnarray}}
\newcommand{\eear}{\end{eqnarray}}
\newcommand{\bean}{\begin{eqnarray*}}
\newcommand{\eean}{\end{eqnarray*}}

\allowdisplaybreaks

\begin{document}
\title{Spectral Gap Computations for Linearized Boltzmann Operators}
%
\author{Chenglong Zhang* and Irene M. Gamba**\\
*The Institute for Computational Engineering and Sciences (ICES)\\
e-mail: chenglongzhng@gmail.com\\
**Department of Mathematics and ICES\\
e-mail: gamba@math.utexas.edu\\
The University of Texas at Austin, Austin, TX 78712, USA\\ 
}

\date{ }

\maketitle
\begin{abstract} The quantitative information on the spectral gaps for the linearized Boltzmann operator is of primary importance on justifying the Boltzmann model and study of relaxation to equilibrium. This work, for the first time, provides numerical evidences on the existence of spectral gaps and corresponding approximate values. The linearized Boltzmann operator is projected onto a Discontinuous Galerkin mesh, resulting in a ``collision matrix". The original spectral gap problem is then approximated by a constrained minimization problem, with objective function being the Rayleigh quotient of the ``collision matrix" and with constraints being the conservation laws. A conservation correction then applies. We also showed the convergence of the approximate Rayleigh quotient to the real spectral gap for the case of integrable angular cross-sections. Some distributed eigen-solvers and hybrid OpenMP and MPI parallel computing are implemented.
Numerical results on integrable as well as non-integrable angular cross-sections are provided.
\end{abstract}
%
%
\maketitle
\section{Introduction}

The Boltzmann equation is of primary importance in rarefied gas dynamics and has been the keystone of kinetic theories.
The classical Boltzmann equation models interactions or collisions through a bilinear collision operator, where the collisional kernel models the intramolecular potentials and angular scattering mechanisms known as the angular cross-section.
These intramolecular potentials model from hard spheres to soft potentials up to Coulombic interactions (important for plasma collisional modeling). The scattering angular function models the anisotropic nature of the interactions. The angular cross sections could be integrable (e.g. Grad cutoff kernels) or non-integrable (e.g. Grad non-cutoff kernels).

The existence of solutions and regularity theory of the Boltzmann equation in the space inhomogeneous setting have been great mathematical challenges and still remain elusive. Nevertheless, it is well understood that these qualitative properties depend on the intermolecular potential and the integrability properties of the angular cross-sections. Indeed, the relaxation to equilibrium has been at the core of kinetic theory ever since the works of Boltzmann. It provides an analytic basis for the second principle of thermodynamics for a statistical physics model of a gas out of equilibrium. The well-known Boltzmann's \textbf{H} theorem \cite{Cercignani_BEnApplications} shows the possible convergence process and equilibrium states.

However, it's not enough to justify the Boltzmann model with only non-constructive arguments. It is crucial to obtain quantitative information on the time scale of the convergence. The question of obtaining explicit decay rates in recent new energy methods \cite{Guo1,Guo2,Guo3,Guo4,StrainGuo1,StrainGuo2,Strain_VMB} also motivates the work on studying spectral gaps and coercivity estimates.
Many authors, for instance \cite{Hilbert, Arkeryd_SG, Carleman_SG, Grad_Asymp, Caflisch_softSG, ChangUhlenbeck, MouhotStrain_SG, BarangerMouhot_SG, BobylevCercignani_SG, GolsePoupaud, Pao}, have made enormous efforts on (non-)constructive estimates for the rate of convergence (we refer to \cite{DMV_CelebratingSG} for a review),
among which \emph{Cercignani's conjecture} \cite{Cercignani_Conjecture} is a great inspiration:

For any $f$ and its associated Maxwellian $\mu$,  there is an entropy-entropy production relation

\begin{equation*}
\mathcal{D}(f) \geq \lambda \left[ \mathcal{H}(f) - \mathcal{H}(\mu)\right]
\end{equation*}
where $\mathcal{H}(f)=\int f\log(f)dv$ is the (opposite) entropy; $ \mathcal{D}(f)=-\frac{d}{dt}\mathcal{H}(f)$ is the dissipation of the entropy, or ``entropy-production" functional; $\lambda>0$ is some ``suitable constant". This is actually claiming an exponential convergence towards equilibrium.

In the regime very close to equilibrium, the linearized part of the model can actually dominate. The linearized counterpart of Cercignani's conjecture writes

\begin{equation*}
D(F) \geq \lambda \Vert F-\mathbf{P}F\Vert^{2}_{2}
\end{equation*}
where $D(F)=\langle LF,F\rangle$ is the Dirichlet form of the linearized Boltzmann operator $L$, whose definitions will be introduced later; $\mathbf{P}$ is the orthogonal projection in $L^{2}$ onto the null space $\mathcal{N}(L)$. 

The explicit rate $\lambda$ (if exists) will be our goal. It has been shown that \cite{Carlen1,Carlen2}, for Maxwellian potentials, solutions for spatial homogeneous Boltzmann equation converge to equilibrium exponentially if and only if the initial datum has finite moments of order greater than 2, and with additional moments and smoothness assumptions on the initial datum, the convergence rate will be governed by $\lambda$. Similar properties also hold for hard potentials with angular cut-off \cite{Lu_Mouhot}. While for soft potentials, such expliicit rate doesn't exists and thus the solutions for homogeneous Boltzmann equations won't enjoy an exponential decay \cite{Caflisch_softSG}. There are very limited amount of results on the estimates, and we haven't seen any numerical results that provide the ``exact" rate governing the exponential decay to equilibrium. This will be the first attempt on this problem.

This paper is organized as follows. Section \ref{sec:BE} will provide some preliminaries about the Boltzmann equations and linearized Boltzmann operators; Section \ref{sec:spectralgap} defines and explains the existence of spectral gaps for Boltzmann models with various types of intramolecular potentials
with integrable and non-integrable angular cross-sections. Section \ref{sec:NumericalFormulation} introduces a way to discretize the linearized Boltzmann operator based on Discontinuous Galerkin scheme, which results in the approximate Rayleigh quotient. The minimal Rayleigh quotient will be found outside the null space of the linearized operator.
The convergence of such approximate Rayleigh quotients to the real spectral gaps is also studied analytically. Finally, some numerical results are given in Section \ref{sec:NumericalResults}.

\section{The Boltzmann equations}\label{sec:BE}
The full Boltzmann transport equation is an integro-differential transport equation, with the solution a phase probability density distribution $f(x,v,t) \in \Omega_x \times \mathbb{R}^{d_{v}}\times \mathbb{R}^{+}$ (where $\Omega_x \subseteq \mathbb{R}^{d_{x}}$) measuring the likelihood to find molecules at a location $x$ with molecular velocities $v$ at a given time $t$. Here, we are only interested in the spatially homogeneous Boltzmann equation in $d$-dimensional velocity space, which reads
\begin{eqnarray}
\label{fullBTE}
\frac{\partial f(v, t)}{\partial t} &=& Q_{sym}(f,f)(v,t) \\
f(v, 0)  &=&  f_0(v) \, . \nonumber \\
\end{eqnarray}
Here, the right-hand side symmetrized Boltzmann bilinear operator reads
\begin{eqnarray}
\label{sym_boltz}
    Q_{sym}(f,g)(v)=\frac{1}{2}\int_{\mathbb{R}^{d}}\int_{\mathbb{S}^{d-1}}(f'g'_{*}+f'_{*}g'-fg_{*}-f_{*}g)B(|u|,\sigma)d\sigma dv_{*} \, ,
\end{eqnarray}
where $\mathbb{S}^{d-1}$ is the $(d-1)$-dimensional sphere. Here and in the following, for simplicity, denote $f'=f(v')$, $f'_{*}=f(v'_{*})$ and $f_{*}=f(v_{*})$, with $v'$ and $v'_{*}$ being post-collisional velocities. We will drop the subscript on $Q_{sym}$ and simply write $Q(f,f)$ when $f=g$.
The integration is parameterized in terms of the center of mass and relative velocity.
And on the $d - 1$ dimensional sphere, integration is done with respect to the unit direction given by the elastic post collisional relative velocity.

The elastic law for pre- and post-collisional velocities obeys
\begin{equation}
\label{velocity}
u=v-v_{*}, \ \ \ \ v'= v + \frac{1}{2}(|u|\sigma - u), \  \ \ \  v'_{*} = v_{*} - \frac{1}{2}(|u|\sigma - u) \, .
\end{equation}
Under certain physical backgrounds, the \emph{collision kernel} is assumed to have a product form
\begin{equation}
\label{coll_kernel}
B(|u|, \sigma)=|u|^{\gamma}b(\cos(\theta)), \qquad \gamma \in (-d,1]\, ,
\end{equation}
with \emph{angular cross-sections}
\begin{equation}
\label{angular_cross}
\cos(\theta)=\frac{u\cdot \sigma}{|u|}\, , \quad b(\cos(\theta))\sim \sin^{-(d-1)-\alpha}(\frac{\theta}{2}) \text{ as } \theta\sim 0  \, , \quad \alpha\in(-\infty,2) \, .
\end{equation}
Without loss of generality, we can assume
\begin{equation}
 b(\cos(\theta))=\frac{1}{2^{d-1}\pi}\sin^{-(d-1)-\alpha}(\frac{\theta}{2}) \, .
\end{equation}


The regularity parameters $\gamma$ and $\alpha$ actually correspond to different types of interactions and different power-law molecular potentials.
For interaction potentials obeying spherical repulsive laws
\begin{equation*}
\phi(r)=r^{-(s-1)}, \quad s\in [2,+\infty)
\end{equation*}
the collision kernel and angular cross-section are explicit for $d=3$, that is, $\gamma=(s-5)/(s-1)$ and $\alpha=2/(s-1)$ (see \cite{Cercignani_BEnApplications}).
As a convention, $-d<\gamma<0$ defines Soft Potentials, $\gamma=0$ is the Maxwell Molecules type interaction, $0<\gamma<1$ describes Variable Hard Potentials and $\gamma=1$ is the classical Hard Sphere model. Also, the angular cross-sections can be of short range or long range, that is, $b(\cos(\theta))$ can be integrable for $\alpha<0$ and non-integrable when $\alpha \geq 0$. In particular, the case $\alpha=2$ and $\gamma=-3$ models the grazing collisions under Coulombian potentials, which deduces to the Fokker-Planck-Landau equation which is a primary model for collisional plasmas.

The weak form for (\ref{sym_boltz}) with $f=g$, or called Maxwell form, after a change of variable $u=v-v_{*}$ is given by
\begin{equation}
\label{maxwell_form}
\int_{\mathbb{R}^{d}} Q(f,f)(v)\phi(v)dv = \int_{_{v,u\in \mathbb{R}^{d}}}\!\! f(v)f(v-u)\int_{_{\sigma \in \mathbb{S}^{d-1}}}\!\!\!\!\![\phi(v')-\phi(v)]B(|u|, \sigma)d\sigma dudv
\end{equation}
which is a \emph{double mixing convolution}. 

In spite of its complicated form, $Q(f,f)$ enjoys many interesting and remarkable properties. Among them, the followings are most fundamental and important \cite{Cercignani_BEnApplications}.

\textbf{Collision invariants and conservation laws. } It's not hard to find that
\begin{equation}
\label{maxwell_form2}
    \int_{\mathbb{R}^{d}} Q(f,f)(v)\phi(v)dv = \frac{1}{2}\int_{\mathbb{R}^{d\times d}}ff_{*} \int_{\mathbb{S}^{d-1}} [\phi + \phi_{*} - \phi' - \phi'_{*}]B(|v-v_{*}|, \sigma)d\sigma dv_{*}dv
\end{equation}
Therefore, one can easily deduce (\ref{maxwell_form2}) is identical to zero if
\begin{equation}
\label{coll_invariant}
\phi + \phi_{*} = \phi' + \phi'_{*}
\end{equation}

It's not difficult to think of some prototypical $\phi(v)$ that satisfy (\ref{coll_invariant}), e.g. mass, momentum, kinetic energy and/or
their combinations. Fortunately, it's also provable that (\ref{coll_invariant}) holds if and only if $\phi(v)$ is in the space spanned by these moments. We call the $d+2$ test functions $\phi(v)=1, \textbf{v}, |v|^{2}$ \emph{collision invariants}, which correspond to the conservation of mass, momentum and kinetic energy.

\textbf{Entropy dissipation and H theorem.} For any $f(v)>0$, if set $\phi(v)=\log f(v)$, then one can prove the following dissipation of entropy
\begin{equation}
\label{entropy_disp}
\frac{d}{dt}\int_{\mathbb{R}^{d}}f(v)\log f(v) dv = \int_{\mathbb{R}^{d}}Q(f,f)(v)\log f(v)dv \leq 0
\end{equation}

This dissipation relation actually implies one fact that the equilibrium state will be given by a \emph{Maxwellian distribution}
\begin{equation}
\label{Maxwell_dist}
M(v) = \frac{\rho}{(2\pi T)^\frac{d}{2}}\exp(-\frac{|v-\bar{v}|^2}{2T})
\end{equation}
where $\rho$ is the macroscopic density, $\bar{v}$ the macroscopic
velocity and $T$ the macroscopic temperature ($=R\vartheta$ where
$\vartheta$ is the absolute temperature, $R$ is a gas constant).

\section{The Linearized Boltzmann Operators and Spectral Gaps}\label{sec:spectralgap}

Since our interest focuses on the behavior in a regime very close to equilibrium, we consider the perturbation near equilibrium
\begin{equation}\label{perturbation}
    f=\mu+\mu^{\frac{1}{2}}F \, ,
\end{equation}
with $\mu=(2\pi)^{-\frac{d}{2}}e^{-\frac{|v|^{2}}{2}}$ being the
normalized equilibrium with mass 1, momentum 0 and temperature 1. Then the linearization of homogeneous Boltzmann equation gives an equation for the perturbation $F(v)$,
\begin{equation*}
\partial_{t}F = -L(F) -\Gamma(F,F) \, ,
\end{equation*}
where the \emph{linearized Boltzmann collision operator} $L$ writes
\begin{equation}\label{eqn:Def_linBoltzmann_L}
L(F)=-2\mu^{-\frac{1}{2}}Q_{sym}(\mu , \mu^{\frac{1}{2}} F) \, ,
\end{equation}
and the bilinear operator $\Gamma$ writes
\begin{equation*}
\Gamma(F,F)=\int_{\mathbb{R}^{d}}\int_{\mathbb{S}^{d-1}}\mu^{\frac{1}{2}}_{*}\left[FF_{*}-F'F'_{*} \right]B(|v-v_{*}|,\sigma)d\sigma dv_{*} \, ,
\end{equation*}
which will be a negligible term when close to equilibrium.

In order to find a suitable \emph{Dirichlet form} associated to the linearized Boltzmann operator $L$ that allows us to generate a sound Rayleigh Quotient structure, one can perform on \eqref{eqn:Def_linBoltzmann_L} exchanges of coordinates $v \leftrightarrow v_{*}$ and $(v,v_{*})\leftrightarrow (v',v'_{*})$. Note, in the latter case, there is a reversal of direction on $\sigma$ for which the Jacobian of change of coordinates remains 1. We eventually obtain the Dirichlet form that writes
\begin{eqnarray}
\begin{aligned}
\label{Dirichlet}
    &\langle L(F),F\rangle :=-\int_{\mathbb{R}^{d}} 2Q_{sym}(\mu,\mu^{\frac{1}{2}}F)F\mu^{-\frac{1}{2}}(v)dv  \\
    &=\frac{1}{4}\int_{\mathbb{R}^{2d}}\int_{\mathbb{S}^{d-1}}\mu\mu_{*}\left(\frac{F(v')}{\mu^{\frac{1}{2}}(v')}+\frac{F(v'_{*})}{\mu^{\frac{1}{2}}(v'_{*})}-\frac{F(v)}{\mu^{\frac{1}{2}}(v)}-\frac{F(v_{*})}{\mu^{\frac{1}{2}}(v_{*})}\right) ^{2}\\
    & \quad \cdot B(u, \sigma)d\sigma dv_{*} dv \\
    &=-\int_{\mathbb{R}^{2d}}\int_{\mathbb{S}^{d-1}} \left[F(v)\mu^{\frac{1}{2}}(v_{*})+F(v_{*})\mu^{\frac{1}{2}}(v)\right]\left[F(v')\mu^{\frac{1}{2}}(v'_{*})-F(v)\mu^{\frac{1}{2}}(v_{*})\right] \\
    &\quad \cdot B(u, \sigma)d\sigma dv_{*}dv \\
    &=-\int_{\mathbb{R}^{2d}}\int_{\mathbb{S}^{d-1}} \mu(v) \mu(v_{*})\left[g(v)+g(v_{*})\right]\left[g(v')-g(v)\right]B(u,\sigma)d\sigma dv_{*}dv \\
\end{aligned}
\end{eqnarray}
where the second line uses the fact that $\mu\mu_{*}=\mu'\mu'_{*}$ and the last line changes $g(v)=\frac{F(v)}{\mu^{1/2}(v)}$. The linear operator $L$ has basic properties \cite{Cercignani_BEnApplications}:
\begin{itemize}
\item It is an unbounded symmetric (self-adjoint) operator on un-weighted $L^{2}(\mathbb{R}^{d})$;
\item It is a positive operator, i.e has non-negative real spectrum;
\item The null space $F(v)\in \mathcal{N}(L)=\mu^{\frac{1}{2}}\cdot \text{span}\{1,v,|v^{2}|\}$. Thus $0$ is an eigenvalue of multiplicity $d+2$.
\end{itemize}

To study the decay of $F$ for $t\rightarrow\infty$, we need to study the eigenvalue problem
\begin{equation}\label{LB_eigenproblem}
Lg=\lambda g
\end{equation}
for which, we have known it has $d+2$ eigen-solutions (collision invariants) for $\lambda=0$. All the other $\lambda>0$.

If the eigen-solutions of eqn (\ref{LB_eigenproblem}), $g_{\lambda}(v)$, can be taken as generalized functions, then it's known that the linearized Boltzmann equation
\begin{equation}
\partial_{t}F=-LF
\end{equation}
has solutions written as \cite{Cercignani_BEnApplications, Gelfand_genfun}
\begin{equation}
F(v,t)=\int^{\lambda_{\infty}}_{\lambda_{0}}e^{-\lambda t}g_{\lambda}(v)h_{\lambda}(v)d\lambda + \sum^{d+1}_{i=0}h_{i}(v)\phi_{i}(v)
\end{equation}
where $h_{\lambda}(v)$ is an arbitrary function depending on $\lambda$ and the integrals extends to all $\lambda\neq 0$ for which $g_{\lambda}\neq 0$ exists.
If some $\lambda$'s form a discrete set, then the corresponding integral is replaced by the sum $\sum_{k}e^{-\lambda_{k}t}g_{k}(v)h_{k}(v)$. So, if $\lambda_{0}\neq 0$ exists,
$F(v)$ decays exponentially into the null space $\mathcal{N}(L)$; while if $\lambda_{0}=0$, the decay is not exponential and depends on initial datum.

\begin{definition}
[Spectral Gap \cite{Mouhot_SGReview}]\label{def1}
Denote by $\sigma(L)$ the spectrum for the operator $L$. For the case $\sigma(L)\subseteq R^{+}$ (i.e. non-negative spectrum), the spectral gap is
defined as the distance between $0$ and $\sigma(L)\setminus \{0\}$.
\end{definition}

The spectral gap is the solution to the constrained minimization problem:
\begin{equation}
\begin{split}\label{eqn:Minimization}
& \min \quad \frac{\langle L(F),F\rangle}{\parallel F \parallel^{2}_{L_{2}}} \\
& s.t \quad F\perp \mathcal{N}(L)
\end{split}
\end{equation}

It tells us how the entropy production functional (the Dirichlet form) is bounded by the relative entropy and thus gives an estimate on the exponential decay of the solutions to the Boltzmann equation.

Thus both the theoretical as well as numerical existence of this ``spectral gap" is very important to us. We will see in the following that the existences of spectral gaps depend on the types of intermolecular potentials ($\gamma$) as well as the integrability of the angular cross-section ($b(\cos(\theta))$). We will look at them separately.

\subsection{Integrable Angular Cross-section}\label{sec:cutoff}

The study on the spectral properties of the linearized Boltzmann collision operator can be traced back all the way to Hilbert \cite{Hilbert}. He suggested the splitting, in the case of hard spheres, between the local and non-local parts of $L$ and proved the compactness of the non-local part. Then Carleman \cite{Carleman_SG} introduced the use of so-called Weyl's theorem to prove the existence of a spectral gap. Then Grad \cite{Grad_Asymp} generalized it to hard potentials with cutoff ($0<\gamma\leq 1$). Then Caflisch \cite{Caflisch_softSG} and Golse and Poupaud \cite{GolsePoupaud} proved the non-existence of spectral gap for soft potentials with cutoff but the existence of a ``degenerated" spectral gap. All the above results are non-constructive. The first constructive estimates were given by Baranger and Mouhot \cite{BarangerMouhot_SG} for the hard spheres model.

For the integrable angular cross-sections, index $\alpha<0$ in (\ref{angular_cross}). Basically, by splitting, $L$ writes
\begin{equation}\label{L_split}
    L(F)(t,v)=\nu(v)F(t,v)+(\mathbf{K}F)(t,v)
\end{equation}
where the collision frequency
\begin{equation}\label{coll_freq}
  \nu(v)=\int_{\mathbb{R}^{d}}\int_{\mathbb{S}^{d-1}} \mu(v_{*}) B(|v-v_{*}|,\hat{u}\cdot \sigma)d\sigma dv_{*}
\end{equation}
and the integral operator $\mathbf{K}$ with kernel $k(v,\eta)$ can be explicitly given. Here, a remarkable feature is that the non-local $\mathbf{K}$ is a compact bounded integral operator.

According to

\begin{theorem}
[Weyl's]\label{weyl}
The essential spectrum (here, the continuous spectrum due to the self-adjoint $L$) is unchanged under a compact perturbation.
\end{theorem}

We easily get that the information of continuous spectrum is completely contained in the local part $\nu(v)$. If assuming a normalized angular cross-section,
i.e. $ \int_{S^{d-1}}b(\hat{u}\cdot \sigma)d\sigma=1$, then,
\begin{equation}\label{eqn:collisionfrq}
    \nu(v)=(2\pi)^{-\frac{d}{2}}\int_{\mathbb{R}^{d}} |v-v_{*}|^{\gamma}e^{-\frac{|v_{*}|^{2}}{2}}dv_{*}
\end{equation}
\begin{itemize}
\item $\gamma\geq 0$,  which is the hard potential model, we can see the continuous spectrum will range from some positive value to infinity. What's left is the discrete spectrum, i.e the eigenvalues. There will be a smallest positive one, which is the spectral gap;
\item $\gamma< 0$,which is the soft potential model, the continuous spectrum can go all the way down to zero; thus we cannot expect a spectral gap. (But, there will be a ``degenerate" one.)
\end{itemize}

The spectrum can be described with pictures, see Figure \ref{fig:sg_VHP}, Figure \ref{fig:sg_Maxwell} and Figure \ref{fig:sg_soft}.
\begin{figure}[!htb]
\centering
\begin{minipage}[t]{0.45\linewidth}
\centering
\includegraphics[width=60mm]{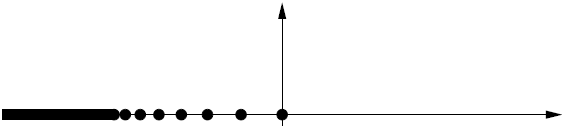}
\caption{Spectrum of $-L$ for variable hard potential with angular cutoff}\label{fig:sg_VHP}
\end{minipage}
\hfill
\begin{minipage}[t]{0.45\linewidth}
\centering
\includegraphics[width=60mm]{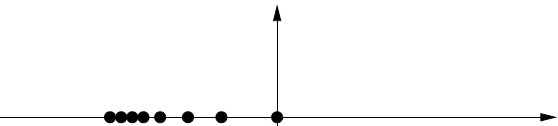}
\caption{Spectrum of $-L$ for Maxwell type with angular cutoff}\label{fig:sg_Maxwell}
\end{minipage}
\hfill
\begin{minipage}[t]{0.45\linewidth}
\centering
\includegraphics[width=60mm]{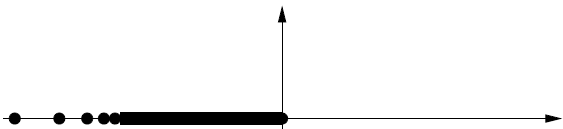}
\caption{Spectrum of $-L$ for soft potential with angular cutoff}\label{fig:sg_soft}
\end{minipage}
\end{figure}

Thus the geometry of the spectrum of linearized Boltzmann operators is clear to us. We will revisit the details of splitting in the next section. A numerical treatment can be designed based on this property of ``splitting".

\subsection{Non-integrable Angular Cross-section} \label{sec:noncutoff_conjecture}

We no longer have the above ``splitting" property with an non-integrable $b(\cos(\theta))$. Thus the above perturbation theories
may no longer directly apply to the spectrum of non-cutoff linearized Boltzmann. However, with a suitable choice of truncated angular domain, which depends on relative velocities, one can still perform some ``splitting" and study each term separately. Thus some constructive coercivity estimates for the Dirichlet form can be found and so for the spectral gaps. This is
what Mouhot \& Strain \cite{MouhotStrain_SG} conjectured and partially proved
{
\begin{theorem}[Mouhot \& Strain]
With the collision kernel $B$ specified in this paper, one has
\begin{itemize}
  \item For any $\epsilon>0$ there is a constructive constant $C_{B,\epsilon}$, such that the Dirichlet form satisfies:
  \begin{equation}\label{}
    \langle LF,F\rangle \geq C_{B,\epsilon}\|(F-\mathbf{P}F)<v>^{\gamma+\alpha-\epsilon}\|^{2}_{L^{2}(R^{d})}.
  \end{equation}
  \item There is a non-constructive constant $C_{B,0}$ such that
  \begin{equation}\label{}
    \langle LF,F\rangle \geq C_{B,0}\|(F-\mathbf{P}F)<v>^{\gamma+\alpha}\|^{2}_{L^{2}(R^{d})}.
  \end{equation}
\end{itemize}
\end{theorem}

where $<v>=(1+|v|^{2})^{\frac{1}{2}}$. Therefore, it is sufficient to claim that when $\gamma+\alpha\geq 0$, there exists a spectral gap for linearized Botlzmann operator. But they
went further and conjectured the necessary part

\begin{cnjctr}[Mouhot \& Strain]
\label{conj:MouhotStrain}
With $\gamma\in(-d,\infty)$ and $\alpha\in[0,2)$ in $B$, the linearized Boltzmann collision operator associated to $B$ admits a spectral gap if and only if $\gamma+\alpha\geq 0$. Moreover this statement is still valid if one includes formally the case of angular cutoff in $``\alpha=0"$, and add the linearized Landau collision operator as the limit case $``\alpha=2"$.
\end{cnjctr}

Recently the necessary part of the conjecture, that is, if the linearized Boltzmann operator observes a spectral gap then $\gamma+\alpha\geq 0$, was answered by Gressman and Strain \cite{GressmanStrain} by proving sharp constructive
upper and lower bounds for the linearized collision operator in terms of a geometric fractional Sobolev norm. In the following session, we will exhibit numerically, if $\gamma+\alpha < 0$, then there exists no spectral gap.

\section{The Discontinuous Galerkin Projections and Approximate Rayleigh Quotients}\label{sec:NumericalFormulation}

In this section, we introduce how to project the original eigenvalue problem onto a finite approximation space, based on Discontinuous Galerkin methods. The key is the treatment of
the angular integrals over the $d-1$ dimensional sphere $\mathbb{S}^{d-1}$. Our DG approximation can handle both integrable and non-integrable angular cross-sections. This is also
the basement of the deterministic DG solvers for fully nonlinear Boltzmann equations, which was also developed by the authors \cite{DGBE_ChenglongGamba}.
Particularly, for operators with integrable angular cross-sections, it can be specially reformulated based on so-called ``Grad splitting'' and can be easily projected onto our DG meshes.

\subsection{Grad Splitting for Integrable Angular Cross-section}
For integrable angular cross-sections, we can easily develop a numerical formulation based on the ``splitting'' property of the operator. Recall the splitting (\ref{L_split}) for $L$.
The collision frequency $\nu(v)$ is well-defined. $\mathbf{K}F$  is given by
\begin{equation}\label{def_K}
\begin{split}
    \mathbf{K}F(v)&=\mu^{\frac{1}{2}}(v)\int_{\mathbb{R}^{d}\times\mathbb{S}^{d-1}}\mu^{\frac{1}{2}}(v_{*})F(v_{*})B(|v-v_{*}|,\hat{u}\cdot \sigma)d\sigma dv_{*} \\
    &-\!\!\int_{_{\mathbb{R}^{d}\times\mathbb{S}^{d-1}} }\!\!\!\!\!\![\mu^{\frac{1}{2}}(v_{*})\mu^{\frac{1}{2}}(v')F(v'_{*})\!+\!\mu^{\frac{1}{2}}(v_{*})\mu^{\frac{1}{2}}(v'_{*})F(v')]\!B(\!|v-v_{*}|,\hat{u}\cdot \sigma\!)d\sigma dv_{*}\\
    &:=\mathbf{K}_{1}F - \mathbf{K}_{2}F
\end{split}
\end{equation}
where one can define the kernel $k_{1}(v,\xi)$ for the integral operator $\mathbf{K}_{1}$
\begin{equation}
 \label{def_K1}
k_{1}(v,\xi)=\mu^{\frac{1}{2}}(v)\mu^{\frac{1}{2}}(\xi)|v-\xi|^{\gamma}\int_{S^{d-1}} b((v-\xi)\cdot\sigma)d\sigma \, .
\end{equation}
The remaining part of (\ref{def_K}) defines $\mathbf{K}_{2}$. The kernel $k_{2}(v,\xi)$ will be derived explicitly.

Let's start from \textit{Carleman Representation}, which is actually transforming the integrals over spheres to integrals over some orthogonal planes.

\begin{lemma}
[Carleman]\label{lem1}
The following identity holds for any appropriate test functions $\phi(z)$:$\mathbb{R}^{d}\rightarrow \mathbb{R} $\\
 \begin{equation}\label{Carleman}
\int_{S^{d-1}}\phi(\frac{|u|\sigma - u}{2})d\sigma = 2^{d-1}|u|^{2-d}\int_{ \mathbb{R}^{d}} \phi(z)\delta(|z|^{2}+z\cdot u)dz
\end{equation}
where $u\in\mathbb{R}^{d}$ is an arbitrary vector and $\delta$ is the one-dimensional Dirac delta function.
\end{lemma}

If we take the following changes of variables
\begin{equation}
 u=v-v_{*} \, ,\quad z=\frac{1}{2}(|u|\sigma-u)\, , \quad w= -\frac{1}{2}(|u|\sigma+u) \,,
\end{equation}
then, $u=-(z+w)$, $v_{*}=v+w+z$, $v'_{*}=v+w$ and $\xi:=v'=v+z$.
Noticing the relationship $|v'-v|=|u|\sin(\theta/2)$ and $|u|=(|\xi-v|^{2}+|w|^{2})^{\frac{1}{2}}$, we obtain the integral form of
$\mathbf{K}_{2}F$ given by
\begin{equation*}
\begin{split}
&\mathbf{K}_{2}F(v) \!
:= \! 2^{d}\!\!\int_{_{\mathbb{R}^{2d}}}\!\!\! \mu^{\frac{1}{2}}(\!v\!+\!w\!+\!z\!)\mu^{\frac{1}{2}}(\!v\!+\!w\!)F(\!v\!+\!z\!) |u|^{2-d}\!B(u,\frac{2z\!+\!u}{|u|})\delta(z\!\cdot\! (\!z\!+\!u\!))dzdu \\
&= 2^{d}\int_{_{\mathbb{R}^{2d}}}\!\! \mu^{\frac{1}{2}}(v+w+z)\mu^{\frac{1}{2}}(v+w)F(v+z)\tilde{B}(w,z)\delta(z\cdot w)dzdw \\
&= 2^{d}\int_{_{\mathbb{R}^{d}\times w\perp z}} |z|^{-1} \mu^{\frac{1}{2}}(v+w+z)\mu^{\frac{1}{2}}(v+w)F(v+z)\tilde{B}(w,z)dzdw \\
&= \frac{2}{\pi} \int_{_{\mathbb{R}^{d}\times w\perp z}} F(\xi)\mu^{\frac{1}{2}}(\xi+w)\mu^{\frac{1}{2}}(v+w)|\xi-v|^{-d-\alpha}\left(|w|^{2}+|\xi-v|^{2}\right)^{\frac{\gamma+1+\alpha}{2}}d\xi dw \, ,
\end{split}
\end{equation*}
where we used the relationship $w\perp z$ and
\begin{equation*}
 \tilde{B}(w,z) = |w+z|^{2-d}B(-(w+z), \frac{z-w}{|z+w|}) = \frac{1}{2^{d-1}\pi}|z|^{-(d-1)-\alpha} \left(|w|^{2}+|z|^{2}\right)^{\frac{\gamma+1+\alpha}{2}} \, .
\end{equation*}
Therefore, the explicit kernel $k_{2}(v,\xi)$ for integral operator $\mathbf{K}_{2}$ can be extracted, which writes
\begin{equation*}
k_{2}(v,\xi)= \frac{2}{\pi}|\xi-v|^{-d-\alpha} \int_{\Pi}\mu^{\frac{1}{2}}(\xi+w)\mu^{\frac{1}{2}}(v+w)\left(|w|^{2}+|\xi-v|^{2}\right)^{\frac{\gamma+1+\alpha}{2}} dw \, ,
\end{equation*}
where the plane $\Pi:=\{w\in\mathbb{R}^{d}: (\xi-v)\cdot w=0\}$.

However, we can simplify more, following tricks from \cite{CercignaniIllnerPulvirenti}. Notice that
\begin{equation}\label{}
    |v+w|^{2}+|\xi+w|^{2}=2|w+\frac{1}{2}(\xi+v)|^{2}+\frac{1}{2}|\xi-v|^{2} \, ,
\end{equation}
and decompose $\frac{1}{2}(\xi+v)$ into parts perpendicular to $\xi-v$ and parallel to $\xi-v$. The projection onto $\xi-v$ is denoted by $\zeta^{\perp}$, which is
\begin{equation}\label{}
    \zeta^{\perp}:= \left(\frac{1}{2}(\xi+v)\cdot \frac{\xi-v}{|\xi-v|}\right)\frac{\xi-v}{|\xi-v|}=\left(\frac{1}{2}\frac{|\xi|^{2}-|v|^{2}}{|\xi-v|}\right)\frac{\xi-v}{|\xi-v|} \, .
\end{equation}
Its orthogonal part, denoted by $\zeta$, is in the same plane as $w$,
\begin{equation}\label{eqn:Def_zeta}
    \zeta : = \frac{1}{2}(\xi+v)-\zeta^{\perp} = \frac{1}{2}(\xi+v)-\left(\frac{1}{2}\frac{|\xi|^{2}-|v|^{2}}{|\xi-v|}\right)\frac{\xi-v}{|\xi-v|} \, .
\end{equation}
Thus, plugging these into $k_{2}$ gives
\begin{equation}\label{def_K2}
\begin{split}
    k_{2}(v,\xi)&=\frac{2}{\pi}(2\pi)^{-\frac{d}{2}}|\xi-v|^{-d-\alpha}\exp(-\frac{1}{8}|\xi-v|^{2}-\frac{1}{8}\frac{(|\xi|^{2}-|v|^{2})^{2}}{|\xi-v|^{2}})\\
    &\cdot\int_{\Pi}\exp(-\frac{|w+\zeta|^{2}}{2})\left(|\xi-v|^{2}+|w|^{2}\right)^{\frac{\gamma+1+\alpha}{2}}dw \, .
\end{split}
\end{equation}
Clearly, $k_{2}(v, \xi)$ is symmetric.

\emph{Remark. } The kernel $k_{2}(v,\xi)$ can be further simplified if $\gamma+1+\alpha=0$. For example, in the case of 2-d Maxwell model or 3-d hard sphere model, since $\zeta$ is just a shift of $w$ on plane $\Pi$ and thus the integrations on plane $\Pi$ can be done analytically,
\begin{equation}
  k_{2}(v,\xi)=2^{\frac{1}{2}}\pi^{-\frac{3}{2}}|\xi-v|^{-(d-1)-1-\alpha}\exp(-\frac{1}{8}|\xi-v|^{2}-\frac{1}{8}\frac{(|\xi|^{2}-|v|^{2})^{2}}{|\xi-v|^{2}})\, .
\end{equation}
Thus,
\begin{equation}\label{eqn:GradSplitting}
    L(F)(v)=\nu(v)F(v) + \mathbf{K}F(v) \, ,
\end{equation}
where the kernel for the integral operator $\mathbf{K}$ is explicitly given
\begin{equation}\label{}
\begin{split}
    k(v,\xi)&=k_{1}(v,\xi)-k_{2}(v,\xi) \\
    &=(2\pi)^{-\frac{d}{2}}\exp(-\frac{|v|^{2}+|\xi|^{2}}{4})|\xi-v|^{\gamma}\int_{S^{d-1}}b(\sigma)d\sigma \\
    &-\frac{2}{\pi}(2\pi)^{-\frac{d}{2}}|\xi-v|^{-d-\alpha}\exp(-\frac{1}{8}|\xi-v|^{2}-\frac{1}{8}\frac{(|\xi|^{2}-|v|^{2})^{2}}{|\xi-v|^{2}})\\
    &\cdot\int_{\Pi}\exp(-\frac{|w+\zeta|^{2}}{2})\left(|\xi-v|^{2}+|w|^{2}\right)^{\frac{\gamma+1+\alpha}{2}}dw \, ,
\end{split}
\end{equation}
which makes the integral operator $\mathbf{K}$ compact on $L^{2}(\mathbb{R}^{d})$.

Combining (\ref{def_K1}) and (\ref{def_K2}) yields the explicit definition for $\mathbf{K}$, which can be proven to be a Hilbert-Schmidt integral operator (this needs $\alpha<0$ which is satisfied due to the integrability of angular cross-sections), and thus $k_{1}(v,\xi)-k_{2}(v,\xi)$ is
$L^{2}$ integrable. Starting from Carleman representation, we actually have recovered the results from \textit{Grad splitting} \cite{Grad_Asymp}.

\subsection{Discontinuous Galerkin Projections}

Albeit the high complexity of DG discretizations, we still prefer DG approximations because with little knowledge of the behaviors of eigenfunctions, DG approximations are expected to
accommodate various kinds of regular and/or irregular eigenfunctions and thus to provide more accurate eigenvalues. To apply DG, we first need to build a reasonable truncated domain.

\subsubsection{Domain and Mesh}

Let's recall the Dirichlet form (\ref{Dirichlet}) for the linearized Boltzmann operator $L$ and the equivalent minimization problem for the spectral gap (\ref{eqn:Minimization}). If we employ change of variables, $g(v)=\frac{F(v)}{\mu^{1/2}(v)}$,
then, equivalently, the spectral gap problem becomes
\begin{equation}
\begin{split}\label{eqn:Minimization_g}
& \min \quad \frac{\langle L(F),F\rangle}{\parallel g \parallel^{2}_{L_{2}(\mu)}} \\
& s.t \quad g\perp \left(\mu^{-\frac{1}{2}}\cdot \mathcal{N}(L)\right)
\end{split}
\end{equation}
where $\parallel \cdot \parallel_{L_{2}(\mu)}$ is the weighted $L^{2}$ norm with weight $\mu(\cdot)$.

It's not difficult to observe that, $g(v)$ can be restricted onto a truncated domain, $\Omega_{v}=[-V,V)^{d}$, which is large enough such that the objective function and constraint in (\ref{eqn:Minimization_g}) will only
differ than their real values within small errors, respectively. Besides, since the whole linearization only makes sense at the regime very close to equilibrium,
it's still reasonable only consider perturbations $F(v)$ with the same ``compact support'' as $\mu(v)$. Thus, in the following, our computing domain is the truncated $\Omega_{v}$, for $g(v)$ and/or $F(v)$.

\emph{Remark. } It's vitally important to pay attention to the domain truncation here. With a velocity cutoff, we are actually dealing with the corresponding cutoff operator
\begin{equation}
L_{\Omega}=\chi_{\Omega}L
\end{equation}
which will definitely possess a spectral gap due to the finite integration domain. Though, see (\ref{eqn:collisionfrq}) and analysis below for example, this will not essentially influence the spectral gap for $\gamma\geq 0$,
yet for soft potential case, $\chi_{\Omega}L$ is expected to have a ``spectral gap" bounded by $\chi_{\Omega}\mu(v)$, up to some constant factors.
However, as $\Omega$ gets larger, we can expect this ``spectral gap" goes to zero. An analytical reasoning is provided in the convervence analysis.

A regular mesh is applied, that is, we divide each direction into $N$ disjoint elements uniformly, such that $[-L,L]=\bigcup_{k}I_{k}$,
where interval $I_{k}=[w_{k-\frac{1}{2}}, w_{k+\frac{1}{2}})$, $w_{k}=-L+(k+\frac{1}{2})\Delta v$, $\Delta v = \frac{2L}{n}$, $k=0 \ldots n-1$
and thus there is a Cartesian partitioning $\mathcal{T}_{h}=\bigcup_{k}E_{k}$, with uniform cubic element $E_{k}=I_{k_{1}}\otimes I_{k_{2}} ...\otimes I_{k_{d}}$, $k=(k_{1}, k_{2},..., k_{d})$.

Discontinuous Galerkin methods assume piecewisely defined basis functions, that is
\begin{equation}
 \label{pw_def}
g(v)=\sum_{k}\textbf{u}_{k}\cdot\Phi(v)\chi_{k}(v)
\end{equation}
where multi-index $k=(k_{1}, k_{2},..., k_{d})$, $0\leq |k| < (n-1)^{3}$; $\chi_{k}(v)$ is the characteristic function over element $E_{k}$; coefficient vector $\textbf{u}_{k}=(\textbf{u}^{0}_{k},..., \textbf{u}^{p}_{k})$,
where $p$ is the total number of basis functions locally defined on $E_{k}$; basis vector $\Phi(v)=(\phi_{0}(v),..., \phi_{p}(v))$.
Usually, we choose element of basis vector $\Phi(v)$ as local polynomial in $P^{p}(E_{k})$, which is the set of polynomials of total degree at most $p$ on $E_{k}$.
For sake of convenience, we select the basis such that $\{ \phi_{i}(v): i=0,...,p\}$ are orthogonal. For example, when $d=3$, $p=1$, local linear basis over element $E_{k}$ can be set as
\begin{equation}\label{eqn:linearbasis}
 \{1, \frac{v_{1}-w_{k_{1}}}{\Delta v}, \frac{v_{2}-w_{k_{2}}}{\Delta v}, \frac{v_{3}-w_{k_{3}}}{\Delta v}\}.
\end{equation}

\subsubsection{Evaluations of Collision Integrals}

For Boltzmann-type equations, the treatment of various collision kernels always remains the most important and challenging part. To demonstrate our scheme, for simplicity, we take piecewise constant
basis functions as example, i.e. $p=0$, only the characteristic function $\chi_{k}(v)$ is applied over each element $E_{k}$.
Due to the possible singularity in angular cross-section, $b(\cos\theta)$, we keep the ``gain-loss" term and will show that this is where the cancellation of singularity occurs. The following techniques have been applied in the development of conservative DG solvers for homogeneous Boltzmann equations \cite{DGBE_ChenglongGamba}. For completeness, here we will describe again.

Plugging (\ref{pw_def}) back into the Dirichlet form (\ref{Dirichlet}) (the last line of formulas) gives, with change of variables $(v,u)\leftarrow (v, v_{*})$, where $u=v-v_{*}$ is the relative velocity,
\begin{equation}\label{1}
    \langle L(F),F\rangle = \mathbf{u}^{T}\mathbf{G}\mathbf{u}
\end{equation}
with $\mathbf{G}$ the ``collision matrix'' with $N\times N$ blocks, each of which is $(p+1)^{d}\times (p+1)^{d}$ block defined as
\begin{equation}\label{G_entry}
\begin{split}
\mathbf{G}(k,m)&=\int_{\mathbb{R}^{d}}\int_{\mathbb{R}^{d}} \mu(v)\mu(v-u) \left(\Phi(v)\chi_{k}(v) + \Phi(v-u)\chi_{k}(v-u)\right) \\
&\otimes  \int_{\mathbb{S}^{d-1}}\left(\Phi(v')\chi_{m}(v') -\Phi(v)\chi_{m}(v)\right)B(u, \sigma)d\sigma dudv
\end{split}
\end{equation}
Let's only look at a generic term
\begin{equation*}
\begin{split}
&\int_{\mathbb{R}^{d}}\int_{\mathbb{R}^{d}} \mu(v)\mu(v-u) \chi_{k}(v)  \int_{\mathbb{S}^{d-1}}\left(\phi_{i}(v')\chi_{m}(v') -\phi_{i}(v)\chi_{m}(v)\right)B(u, \sigma)d\sigma dudv \\
&=\sum_{\bar{k}}\int_{v\in E_{k}}\int_{v-u \in E_{\bar{k}}}\mu(v)\mu(v-u) \int_{_{\mathbb{S}^{d-1}}}\!\!\!\left(\phi_{i}(v')\chi_{m}(v') -\phi_{i}(v)\chi_{m}(v)\right)B(u, \sigma)d\sigma dudv
\end{split}
\end{equation*}
The other terms are evaluated in a same way.

Due to the convolution formulation, the integrals w.r.t $v,u$ can be approximated through \emph{Triangular quadratures}. Indeed, along each dimension, if $v_{i}\in I_{k_{i}} $, $v_{i}-u_{i}\in I_{\bar{k}_{i}}$, then
$(v_{i}, u_{i})$ will form a parallelogram which can be divided into two triangles. See Figure \ref{fig:vu_parallelogram}.

\begin{figure}[!htb]
  \centering
  \includegraphics[width=50mm]{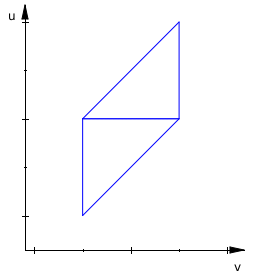}\\
  \caption{Along each dimension, $(v_{i}, u_{i})$ forms two right triangles}\label{fig:vu_parallelogram}
\end{figure}

The integrals on the sphere take the most efforts, because one has to figure out how the Cartesian cubes intersect with the spheres. Let's extract the angular integrals in (\ref{G_entry}), denoted by $g_{m,i}(v,u)$, and study it separately
\begin{equation}\label{AngInt_g}
\mathfrak{g}_{m,i}(v,u):=\int_{\mathbb{S}^{d-1}}\left(\phi_{i}(v')\chi_{m}(v') -\phi_{i}(v)\chi_{m}(v)\right)b(\frac{u\cdot \sigma}{|u|})d\sigma
\end{equation}
The treatments for (\ref{AngInt_g}) follows exactly the same as in the work \cite{DGBE_ChenglongGamba} where deterministic DG solvers for nonlinear Boltzmann equations are developed. Here, we
restate them below.

For any fixed $v,u$, the post-collisional velocity $v'$ will be on the surface of a ball centered at $v-\frac{u}{2}$ with radius $\frac{|u|}{2}$.
The angular cross-section $b(\cos\theta)$ itself may contain non-integrable singularity at $\theta=0$. However, the ``gain-loss" terms in the above square bracket will absorb
the singularity in $b(\cos\theta)$ and make it integrable. Our scheme has to take this issue into account and design a careful way of computing.

\begin{itemize}
  \item[1.]\    Integrable $b(\cos\theta)$.

This case allows to split the ``gain'' and ``loss'' terms. Only ``gain'' terms involve post-collisional velocity $v'$ and can be studied separately.

For $d=2$, the angular integrals (\ref{AngInt_g}) can be evaluated analytically. Indeed, for fixed $v,u$, the regions over the cycle $\sigma=(\sin\theta, \cos\theta)$ such that $v'=v+0.5(|u|\sigma-u)\in E_{m}$
can be exactly figured out, by solving a system of trigonometric inequalities
\begin{equation}
 \left\{ \begin{aligned}
         v_{1}+0.5(|u|\sin\theta-u_{1}) &\in I_{m_{1}} \\
         v_{2}+0.5(|u|\cos\theta-u_{2}) &\in I_{m_{2}}
\end{aligned} \right.
\end{equation}
We have built a programmable routine of deriving all possible overlapped intervals of $\theta$.

The case $d=3$ performs similarly. We solve the following nonlinear trigonometric inequalities
\begin{equation}
 \left\{ \begin{aligned}
 v_{1}+0.5(|u|\sin\theta\cos\varphi-u_{1}) &\in I_{m_{1}} \\
 v_{2}+0.5(|u|\sin\theta\sin\varphi-u_{2}) &\in I_{m_{2}} \\
 v_{3}+0.5(|u|\cos\theta-u_{3}) &\in I_{m_{3}}
\end{aligned} \right.
\end{equation}

The third inequality will give a range for the polar angle $\theta$, and all integrals w.r.t $\theta$ will be performed by adaptive quadratures, say, CQUAD in GSL \cite{GSL};
for any fixed $\theta$, the first two inequalities will decide the range of azimuthal angle
$\varphi$ exactly (by invoking the routine mentioned above).

\emph{Note: }The angle $\theta$ above is NOT the one defined in (\ref{angular_cross}).


\item[2.]\  Non-integrable $b(\cos\theta)$.

Consider a local spherical coordinate system with $u$ being the polar direction. Then, consider a transformation which rotates the polar direction back
onto $z$-axis of the Cartesian coordinate system. The orthogonal rotation matrix $A$ can be constructed explicitly

$d=2$:
\begin{equation}\label{RotMat}
A=\frac{1}{|u|}\left(
                 \begin{array}{cc}
                    -u_{2} & u_{1}\\
                     u_{1} & u_{2} \\
                 \end{array}
               \right)
\end{equation}

$d=3$:
\begin{equation}\label{RotMat}
A=\frac{1}{|u|}\left(
                 \begin{array}{ccc}
                   \frac{u_{1}u_{3}}{\sqrt{u^{2}_{1}+u^{2}_{2}}} &  \frac{u_{2}u_{3}}{\sqrt{u^{2}_{1}+u^{2}_{2}}} & -\sqrt{u^{2}_{1}+u^{2}_{2}} \\
                    -\frac{u_{2}|u|}{\sqrt{u^{2}_{1}+u^{2}_{2}}} & \frac{u_{1}|u|}{\sqrt{u^{2}_{1}+u^{2}_{2}}}  & 0\\
                     u_{1} & u_{2} & u_{3} \\
                 \end{array}
               \right)
\end{equation}
where we assume $u^{2}_{1}+u^{2}_{2}\neq 0$, otherwise, the rotation matrix is reduced to the identity matrix.

Then, consider a change of variable $\sigma\leftarrow A^{-1}\sigma=A^{T}\sigma$, for which the Jocobian is $1$. If denote by $\theta$ the angle between $u$ and $\sigma$ , as exactly defined in (\ref{angular_cross}),
recalling post-collisional velocity $v'=v+\frac{1}{2}(|u|\sigma-u)$, we have
\begin{equation}\label{gmi}
\begin{split}
\mathfrak{g}_{m,i}(v,u)&=\int_{\mathbb{S}^{d-1}}\left[\phi_{i}\circ\chi_{m}(v+z)-\phi_{i}\circ\chi_{m}(v)\right]b(\cos \theta)d\sigma \nonumber \\
&=\int_{\mathbb{S}^{d-1}}\left[\phi_{i}\circ\chi_{m}(v-\frac{u}{2}+\frac{|u|}{2}\sigma)-\phi_{i}\circ\chi_{m}(v)\right]b(\cos \theta)d\sigma
\end{split}
\end{equation}
where, if $d=2$: $z=\frac{|u|}{2}A^{T}\left(\sin\theta, \cos\theta-1 \right)^{T}\!,$ $\sigma=A^{T}\left(\sin\theta, \cos\theta \right)^{T}$.

If $d=3$: the variable  $z$ and $\sigma$ are written in spherical coordinate system, given by $z= \! \frac12 (|u|\sigma-u) \!=\! \frac{|u|}{2}A^{T}\left(\sin \theta \cos \varphi, \sin \theta \sin \varphi, \cos \theta -1\right)^{T}$ and
the scattering direction $\sigma=A^{T}\left(\sin \theta \cos \varphi, \sin \theta \sin \varphi,  \cos \theta\right)^{T}$.

We take $d=3$ for example. The whole domain of $(\theta, \varphi)$, i.e. the sphere, can be divided into the following four subdomains: (1) $S_{1}=[0, \theta_{0}]\times [0, 2\pi]$;
(2) $S_{2}=[\theta_{0}, \theta_{1}]\times I_{\varphi}(\theta)$; (3) $S_{3}=[\theta_{0}, \theta_{1}]\times \left([0,2\pi]\setminus I_{\varphi}(\theta) \right)$;
and (4) $S_{4}=[\theta_{1}, \pi]\times [0,2\pi]$. Here $\theta_{0}$ is determined according to the following policy: when $v\in E_{m}$, $\sin \frac{\theta_{0}}{2} = \min(1, \frac{1}{|u|}\text{dist}(v, \partial E_{m}))$ by noticing that $|z|=|u|\sin\frac{\theta}{2}$; when $v \notin E_{m}$,
$\theta_{0}$ is the smallest possible $\theta$ such that $v'$ lies in $E_{m}$. $\theta_{1}$ is
the largest possible $\theta$ such that $v'$ lies in $E_{m}$. $I_{\varphi}(\theta)$ are effective intervals for $\varphi$, depending on $\theta$, such that
$v'$ lies in $E_{m}$.

Due to the characteristic functions in the integrands of $g_{m,i}(v,u)$ (\ref{gmi}), we have the following four cases
\begin{itemize}
  \item[(a)]\  `0-0': when $v'\notin E_{m}$ and $v\notin E_{m}$. It's trivial because it contributes nothing to the final weight matrix.
  \item[(b)]\  `1-0': when $v'\in E_{m}$ but $v\notin E_{m}$. In this case, the effective domain (where $g_{m,i}(v,u) \neq 0$) is $(\theta, \varphi)\in S_{2}$
  \begin{equation*}
  \mathfrak{g}_{m,i}(v,u) = \int_{S_{2}}\phi_{i}(v')b(\cos\theta)\sin\theta d\theta d\varphi
  \end{equation*}
  \item[(c)]\ `0-1': when $v'\notin E_{m}$ but $v\in E_{m}$. In this case, the effective domain is $(\theta, \varphi)\in S_{3}\cup S_{4}$.
   \begin{equation*}
  \mathfrak{g}_{m,i}(v,u) = -\int_{S_{3}\cup S_{4}}\phi_{i}(v)b(\cos\theta)\sin\theta d\theta d\varphi
  \end{equation*}
  \item[(d)]\  `1-1': when $v'\in E_{m}$ and $v\in E_{m}$. In this case, the effective domain is $(\theta, \varphi)\in S_{1}\cup S_{2}$.
  \begin{equation*}
  \mathfrak{g}_{m,i}(v,u)=\int_{S_{1}\cup S_{2}}\left[\phi_{i}(v')-\phi_{i}(v)\right]b(\cos \theta)\sin\theta  d\varphi d\theta
  \end{equation*}
    \end{itemize}
  
  We have to pay special attention to integrals over $S_{1}$, where the singularity is absorbed.
  Recall $\phi_{i}(v)$ are polynomial basis locally defined on each element $E_{m}$ (if it's piecewise constants, then this case becomes trivial) and $v'=v+z$.
  Since $z\sim 0$, we take the Taylor expansion of $\phi_{i}(v')$ around $v$,
  \begin{equation*}
  \phi_{i}(v')-\phi_{i}(v) = \nabla \phi_{i}(v) \cdot z + \frac{1}{2}z^{T}\nabla^{2}\phi_{i}(v)z + O(|z|^{3}).
  \end{equation*}

  So, it's not hard to observe that, for terms with lowest power of $\sin\theta$,  the azimuthal angle $\varphi$ will be integrated out and leaves only powers of $1-\cos\theta$, which will help cancel the
  singularity in $b(\cos\theta)$. That is,
  \begin{equation}
  \begin{split}
  &\int^{\theta_{0}}_{0}\int^{2\pi}_{0}\left[\phi_{i}(v')-\phi_{i}(v)\right]b(\cos \theta)\sin\theta  d\varphi d\theta\\
  &\leq C\int^{\theta_{0}}_{0} (1-\cos\theta)\sin^{-2-\alpha} \frac{\theta}{2} \sin \theta  d\varphi d\theta \nonumber \\
  &\leq C\int^{t_{0}}_{0} t^{1-\alpha} dt \quad (\text{ change } t=\sin\frac{\theta}{2}, \quad t_{0}=\sin\frac{\theta_{0}}{2}) \nonumber \\
  &= \frac{C}{2-\alpha}t^{2-\alpha}_{0} \quad (\text{ notice } \alpha<2)
  \end{split}
  \end{equation}

In practice, the sets $S_{1}$ and $S_{2}$ can be combined. The outer integration w.r.t the polar angle $\theta$ is performed using adaptive quadratures , say CQUAD in GSL \cite{GSL},
and the inner integration
w.r.t $\varphi$ is done analytically by calling a similar routine that derives all possible intervals of $\varphi$.

\emph{Remark. }In practice, the above routine can be only applied to the case when $v,v'$ fall onto the same mesh element (when collision is almost grazing); for other cases, the angular cross-sections can be regarded as integrable (far away from grazing collisions) and thus can call routines in ``Integrable $b(\cos\theta)$".

\end{itemize}

Once $\mathfrak{g}_{m,i}(v,u)$ is done, plugging it back into (\ref{G_entry}), we get the ``collision matrix'' $\mathbf{G}$.

Finally, we would like to mention that, specially for the Grad splitting formulations, the block $\mathbf{G}(k,m)$ can be written out immediately, from
(\ref{eqn:GradSplitting}),
\begin{equation}\begin{split}
\mathbf{G}(k,m) &= (\text{Diagonal block})\int_{E_{k}}\nu(v)\Phi(v)\otimes\Phi(v)dv \\
&\ + \int_{E_{k}}\int_{E_{m}}\left(k_{1}(v,\xi)-k_{2}(v,\xi)\right)\Phi(v)\otimes\Phi(\xi)dvd\xi
\end{split}\end{equation}
which results in a symmetric semi-positive definite collision matrix $\mathbf{G}$.

\subsubsection{The Approximate Rayleigh Quotient}

Recall the equivalent minimization problem for solving spectral gaps in (\ref{eqn:Minimization}) or (\ref{eqn:Minimization_g}).
With the approximation above, we can easily rewrite this constrained minimization problem as
\begin{equation}\label{Rayleigh}
\begin{split}
& \min \quad \frac{\mathbf{u}^{T}\mathbf{G}\mathbf{u}}{\mathbf{u}^{T}\mathbf{D}\mathbf{u}} \\
& s.t \quad \mathbf{C}\mathbf{u}=\mathbf{0}
\end{split}
\end{equation}
where the block diagonal matrix $\mathbf{D}$ generated from the tensor product of the basis functions;  the constraint matrix $\mathbf{C}$ is of size $(d+2)\times M$ (here $M=N(p+1)^{d}$ is the number of coefficients), obtained from the constraints.
\begin{equation}\label{eqn:constraint}
\int F(v)\mu^{\frac{1}{2}}(v)dv = \int F(v)\mu^{\frac{1}{2}}(v)vdv = \int F(v)\mu^{\frac{1}{2}}(v)|v|^{2}dv =0
\end{equation}

We need to find the global optimization solution. To do this, we first find an orthogonal basis $\mathbf{P}$ for the constraint space
\begin{equation}\label{eqn:constraintMatrix}
\mathcal{P}=\lbrace \mathbf{u}\in \mathbb{R}^{M}:\mathbf{Cu}=0 \rbrace
\end{equation}
This can be done through performing $QR$ factorization for $\mathbf{C}^{T}$, the last $M-(d+2)$ columns will form the orthogonal (actually, orthonormal) basis $\mathbf{P}$, of size $M \times (M-(d+2))$ and $\mathbf{P}^{T}\mathbf{P}=\mathbf{I}_{M-(d+2)}$

Then, the minimization problem becomes
\begin{equation}\label{eqn:generalized_eigen0}
    \min_{0\neq b \in \mathbb{R}^{M-(d+2)}} \quad \frac{b^{T}\mathbf{P}^{T}\mathbf{GP}b}{b^{T}\mathbf{P}^{T}\mathbf{DP}b}
\end{equation}
which is equivalently to find the smallest singular value from the generalized eigenvalue problem
\begin{equation}\label{eqn:generalized_eigen1}
\mathbf{P}^{T}\mathbf{GP}=\lambda\mathbf{P}^{T}\mathbf{DP}.
\end{equation}

In practice, instead of solving (\ref{eqn:generalized_eigen0}) and (\ref{eqn:generalized_eigen1}) which requires extra $QR$ decomposition and matrix multiplications, we find out another way to force the constraints (\ref{eqn:constraint}), which is much more efficient and easier to implement. This is done by perturbing the ``collision matrix" $\mathbf{G}$ to its ``$L^{2}$-closest" counterpart, through introducing a ``conservation routine". A similar conservation routine has been successfully applied to deterministic conservative solvers for nonlinear Boltzmann equations based on Spectral methods \cite{GT_jcp} as well as Discontinuous Galerkin methods \cite{DGBE_ChenglongGamba}.

Our objective is to force the eigenvalues to be zeros whenever the functions fall onto the null space $\mathcal{N}(L)$ of operator $L$. That is, to force the conservation, we seek for a perturbation of $\mathbf{Q}:=\mathbf{Gu}$, which is the minimizer of the following constrained optimization problem:

\noindent \textbf{Conservation Routine [Discrete Level]:} Find $\mathbf{Q}_{c}$ (the subscript $c$ means a conservative correction), which is the minimizer of the problem
\begin{equation*}
\begin{split}
&\min \frac{1}{2} (\mathbf{Q}_{c}-\mathbf{Q})^{T} \mathbf{D}(\mathbf{Q}_{c}-\mathbf{Q})  \nonumber \\
&\text{s.t. } \quad \mathbf{C}\mathbf{Q}_{c}=\mathbf{0}
\end{split}
\end{equation*}

Due to the orthogonality of the local basis, $\mathbf{D}$ is a positive definite diagonal matrix with its $j$-th entry $\frac{1}{|E_{k}|}\int_{E_{k}} (\phi_{l}(v))^{2}dv$, $j=(p+1)k+l$. For example, in 3D, when $p=0$, $\mathbf{D}$ is reduced to an identity matrix;
while $p=1$, with the orthogonal basis chosen in (\ref{eqn:linearbasis}),
\begin{equation*}
 \mathbf{D} = \text{Diag }(1, \frac{1}{12}, \frac{1}{12}, \frac{1}{12}, 1, \frac{1}{12},\frac{1}{12},\frac{1}{12},1,...) \, .
\end{equation*}

To solve the minimization problem, we employ the Lagrange multiplier method. Denote by $\lambda\in \mathbb{R}^{d+2}$ the multiplier vector. Then the objective function writes
\begin{equation}
\mathcal{L}(\mathbf{Q}_{c}, \lambda)=\frac{1}{2} (\mathbf{Q}_{c}-\mathbf{Q})^{T} \mathbf{D}(\mathbf{Q}_{c}-\mathbf{Q})-\lambda^{T}\mathbf{C}\mathbf{Q}_{c} \, .
\end{equation}
We can solve it by finding the critical value of $\mathcal{L}$ gives
\begin{equation*}
\left\{ \begin{aligned}
         \frac{\partial \mathcal{L}}{\partial \mathbf{Q_{c}}} &= \mathbf{0} \\
         \frac{\partial \mathcal{L}}{\partial \lambda}&=\mathbf{0}
\end{aligned} \right.
\Longrightarrow
\left\{ \begin{aligned}
&\mathbf{Q}_{c}= \mathbf{Q} + \mathbf{D}^{-1}\mathbf{C}^{T}\lambda \\
         &\mathbf{C}\mathbf{Q}_{c}=\mathbf{0}
\end{aligned} \right.
\Longrightarrow
\lambda=-(\mathbf{C}\mathbf{D}^{-1}\mathbf{C}^{T})^{-1}\mathbf{C}\mathbf{Q}
\end{equation*}
(Here, notice that $\mathbf{C}\mathbf{D}^{-1}\mathbf{C}^{T}$ is symmetric and positive definite and hence exists the inverse.)

Thus, we get the minimizer $\mathbf{Q}_{c}$
\begin{equation}
\mathbf{Q}_{c}=[\mathbb{I}d-\mathbf{D}^{-1}\mathbf{C}^{T}(\mathbf{C}\mathbf{D}^{-1}\mathbf{C}^{T})^{-1}\mathbf{C}]\mathbf{Q}\, ,
\end{equation}
where $\mathbb{I}d$ is an identity matrix of size $M\times M$. Obviously, $\mathbf{Q}_{c}$ is a perturbation of $\mathbf{Q}$.
Therefore, finally, the perturbed ``collision matrix" $\mathbf{G}$ will be
\begin{equation}
\mathbf{G}_{c}=[\mathbb{I}d-\mathbf{D}^{-1}\mathbf{C}^{T}(\mathbf{C}\mathbf{D}^{-1}\mathbf{C}^{T})^{-1}\mathbf{C}]\mathbf{G}
\end{equation}
which is forced to have $d+2$ zero eigenvalues whenever $\mathbf{u} \not\in \mathcal{P}$ defined in (\ref{eqn:constraintMatrix}).

The $(d+3)$-rd eigenvalue of $\mathbf{G}_{c}$ will be defined as our numerical spectral gap.

\subsubsection{Convergence of The Approximate Rayleigh Quotient}\label{sec:Conv_Rayleigh}

We will prove that the above discrete Rayleigh quotient (\ref{Rayleigh}) will converge to the spectral gap solved from (\ref{eqn:Minimization}). With standard approximation theory, it is not hard to prove for integrable angular cross-sections that, the above discrete Rayleigh quotient (\ref{Rayleigh}) converges to the spectral gap of the original linearized Boltzmann operator. We summarize the results it in the following theorem.
\begin{theorem}[Convergence of Rayleigh Quotients] \label{thm:Conv_RayleighQuot}
For the angular integrable (i.e. $\alpha<0$ in (\ref{angular_cross})) linearized Boltzmann operator,
defined in the Dirichlet form (\ref{Dirichlet}),  with a piecewise polynomial approximation (\ref{pw_def}) for the perturbation $F(v)$,
the spectral gap, denoted by $\lambda(G)$, solved from minimized Rayleigh quotient (\ref{Rayleigh}) approximates the original spectral gap, denoted by $\lambda(L)$, solved from (\ref{eqn:Minimization}), in the following way,

\begin{itemize}
  \item When $\gamma\geq 0$, $|\lambda(L)-\lambda(G)|\lesssim h^{k+1}$ \, ;
  \item When $-d<\gamma<0$, $|\lambda(L)-\lambda(G)|\lesssim h^{k+1}+ e^{-\frac{V^{2}}{2}}$ \, ,
\end{itemize}
where $h=\max_{E\in\mathcal{T}_{h}}\text{diam}(E)$ is the mesh size of the regular triangulation, $k$ is the total degree of polynomials in the piecewise polynomial space $\mathbf{P}^{k}$ and $V$ is the lateral size of the computational domain.
The ``$\lesssim$'' is only upto some constant depending on the truncated computational domain $\Omega=[-V,V]^{d}$ as well as eigenfunctions associated with the spectral gap eigenvalue.
\end{theorem}

\begin{proof}

As shown in the Dirichlet form (\ref{Dirichlet}) of $L$, the eigenvalue zero is corresponding to the conservation laws for mass, momentum and kinetic energy.
Therefore, it is of multiplicity $d+2$, with eigenfunctions $\phi_{0}(v)=\mu^{1/2}(v)$, $\phi_{i}(v)=\mu^{1/2}(v)v_{i}$ for $i=1,..,d$ and $\phi_{d+1}(v)=\mu^{1/2}(v)|v|^{2}$.

Suppose the truncated velocity domain $\Omega=[-V,V]^{d}$ is large enough. We are indeed dealing with the cutoff operator $L_{\Omega}=\chi_{\Omega}L$ applying to $\chi_{\Omega}(v)F(v)$.
That is, the kernel, denoted by $k_{\Omega}$, for cutoff $L_{\Omega}$ is given by
\begin{equation}\label{cutL}
k_{\Omega}=\chi_{\Omega}(v)\nu(v)\delta(v-\xi) + \chi_{\Omega}(v)k(v,\xi) \, ,
\end{equation}
where $\delta(v-\xi)$ is short for $\delta(v_{1}-\xi_{1})\cdot\cdot\cdot\delta(v_{d}-\xi_{d})$, $\nu(v)$ is the collision frequency defined in (\ref{eqn:collisionfrq}) and $k(v,\xi)$ is the kernel for the compact operator $\mathbf{K}$ in (\ref{def_K}).

However, the null space $\cal{N}(L)$ is not invariant under the cutoff. Nevertheless, since $\cal{N}(L)$ is spanned by collision invariants weighted with a Gaussian distribution,
as long as $\Omega$ is large enough, the approximation error due to cutoff can be negligible.
To save trouble on dealing with null space, we consider the modified linear operator $\bar{L}$, with the null space of $L$ removed
\begin{equation}
\bar{L}F=LF+\sum^{d+1}_{i=0}\phi_{i}(F, \phi_{i}) \, ,
\end{equation}
where $(F, \phi_{i})=\int_{\mathbb{R}^{d}} F(v)\phi_{i}(v)dv$. This is to replace the integral kernel $k(v,\xi)$ by
\begin{equation}
\bar{k}(v,\xi)=k(v,\xi)+\sum^{d+1}_{i=0}\phi_{i}(v)\phi_{i}(\xi) \, ,
\end{equation}
which is still $L^{2}(\mathbb{R}^{d})$ integrable. That is, $\bar{L}$ can be still written as collision frequency $\nu(v)$ plus a compact perturbation.

Thus, the minimum Rayleigh quotient of $\bar{L}$ is the expected spectral gap, if exists. That is, $\lambda(L)=\lambda(\bar{L})$.
So, we only need to study the approximations for the Rayleigh quotient of operator $\bar{L}$.

Similarly, we are working with the cutoff operator $\bar{L}_{\Omega}=\chi_{\Omega}\bar{L}$ applying to $\chi_{\Omega}(v)F(v)$.
That is, the kernel $\bar{k}_{\Omega}$ for cutoff $\bar{L}_{\Omega}$ is given by
\begin{equation}\label{cutL}
\bar{k}_{\Omega}=\chi_{\Omega}(v)\nu(v)\delta(v-\xi) + \chi_{\Omega}(v)\bar{k}(v,\xi) \, .
\end{equation}

According to Weyl's theorem, for $\gamma\geq 0$, the spectral gap for the new $\bar{L}$ still exists.
And in the case, the cutoff doesn't change
the minimum of the Rayleigh quotient of $\bar{L}$. So, the spectral gap stays the same, or $\lambda(\bar{L})=\lambda(\bar{L}_{\Omega})$.

While for the case $-d<\gamma<0$,
\begin{equation}
\min_{v\in\Omega}\nu(v) \gtrsim e^{-\frac{V^{2}}{2}}\, ,
\end{equation}
which is the lower bound for the continuum spectrum of $\bar{L}_{\Omega}$. This implies, the spectral gap for the cutoff operator $\bar{L}_{\Omega}$ is no larger than $e^{-\frac{V^{2}}{2}}$ (up to some constant factor), if ever exists.
That is, $|\lambda(\bar{L})-\lambda(\bar{L}_{\Omega})|\lesssim e^{-\frac{V^{2}}{2}}$.

Suppose $\mathcal{T}_{h}$ is a regular Cartesian partition for $\Omega$, with mesh size $h=\max_{E\in \mathcal{T}_h{}}\text{diam}(E)$.
We define the standard $d$-dimensional $L^{2}$ projection $P_{h}: f \mapsto P_{h}f$ by
\begin{equation}\label{L2Proj}
\int_{E}P_{h}f(v)\phi(v)dv = \int_{E}f(v)\phi(v)dv, \quad \forall \phi\in \mathbf{P}^{l}|_{E}
\end{equation}

By Poincare's inequality and Sobolev embedding theorems, we can prove the following approximation theory
\begin{equation}\label{appr_theory}
\begin{split}
&\|f-P_{h}f\|_{L^{2}(\mathcal{T}_{h})}\lesssim h^{q+1}\|f\|_{H^{q+1}(\Omega)}, \quad \forall f\in H^{q+1}(\Omega)\, \nonumber \\
&\|P_{h}f\|_{L^{p}(\mathcal{T}_{h})}\lesssim \|f\|_{L^{p}(\Omega)}, \quad \forall f\in L^{p}(\Omega), \quad 1\leq p \leq \infty
\end{split}
\end{equation}
where $L^{p}$ and $H^{q+1}$ are usual Sobolev spaces and Hilbert spaces, respectively.

For any mesh elements $E_{v}$ and $E_{\xi}$, according to the approximation theories (\ref{appr_theory}), it's not hard to prove the following
\begin{equation}
\|F(v)F(\xi)-P_{h}F(v)P_{h}F(\xi)\|_{L^{2}(E_{v}\times E_{\xi})} \leq h^{k+1}\left(\|F\|_{H^{k+1}(E_{v})}\|F\|_{H^{k+1}(E_{\xi})}\right) \, ,
\end{equation}
where $P_{h}F$ is the $L^{2}$ projection defined in (\ref{L2Proj}).

Then, the Dirichlet form is approximated as follows
\begin{equation}
\begin{split}
&|\langle \bar{L}_{\Omega}F, F\rangle - \langle \bar{L}_{\Omega}(P_{h}F), (P_{h}F)\rangle|  \\
&\leq \sum_{m}\sum_{n}\|\bar{k}_{\Omega}\|_{L^{2}(E_{m}\times E_{n})}\|F(v)F(\xi)-P_{h}F(v)P_{h}F(\xi)\|_{L^{2}(E_{m}\times E_{n})}  \\
&\leq C(\Omega)h^{k+1}\|F\|^{2}_{H^{k+1}(\mathcal{T}_{h})} \, ,
\end{split}
\end{equation}
where $C(\Omega)$ is some constant depending on the truncated domain $\Omega$.

And thus, the Rayleigh quotients have the following estimates
\begin{equation}
\begin{split}
&\left|\frac{\langle \bar{L}_{\Omega}F, F\rangle}{\|F\|^{2}_{L^{2}(\Omega)}} - \frac{\langle \bar{L}_{\Omega}(P_{h}F), (P_{h}F)\rangle}{\|P_{h}F\|^{2}_{L^{2}(\mathcal{T}_{h})}} \right| \\
&=\frac{1}{\|F\|^{2}_{L^{2}(\Omega)}\|P_{h}F\|^{2}_{L^{2}(\mathcal{T}_{h})}}\big(\langle \bar{L}_{\Omega}F, F\rangle\big( \|P_{h}F\|^{2}_{L^{2}(\mathcal{T}_{h})} - \|F\|^{2}_{L^{2}(\Omega)} \big) \\
&\quad + \|F\|^{2}_{L^{2}(\Omega)} \left( \langle \bar{L}_{\Omega}F, F\rangle - \langle \bar{L}_{\Omega}(P_{h}F), (P_{h}F)\rangle \right)  \big) \\
&\leq C(\Omega)h^{k+1}\, ,
\end{split}
\end{equation}
which implies,
\begin{equation}
|\lambda(\bar{L}_{\Omega}) - \lambda(G)| \leq C(\Omega)h^{k+1}\, ,
\end{equation}
where now the generic constant $C(\Omega)$ also depends on the eigenfunction associated with the spectral gap eigenvalue.

Finally, noticing
\begin{equation}
|\lambda(L)-\lambda(G)| \leq |\lambda(\bar{L})-\lambda(\bar{L}_{\Omega})| + |\lambda(\bar{L}_{\Omega}) - \lambda(G)|\, 
\end{equation}
gives our final estimates.
 
\end{proof}
The convergence of Rayleigh Quotients for non-integrable angular cross-sections are more subtle. In the following session, we provide numerical results that may indicate the necessary part of Conjecture \ref{conj:MouhotStrain} is true, i.e. there is no spectral gap if $gamma+alpha < 0$, which is consistent with what Gressman \& Strain \cite{GressmanStrain} theoretically proved.

\section{Numerical Results}\label{sec:NumericalResults}

In this section, we will present some results for $2d$ as well as $3d$ linearized Boltzmann operators with integrable angular cross-sections.

The computing of weight matrix $G$ is parallelized with MPI \cite{MPI}. The matrix will be computed and stored in a way of two-dimensional block cyclic distribution \cite{Scalapack},
on a process grid, as shown in Figure \ref{fig:2d_block_cyclic}

\begin{figure}[!htb]
  \centering
%

\includegraphics[width=70mm]{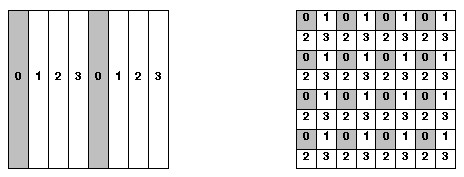}\\
\caption{The 1d block-cyclic column- and 2d block-cyclic distributions}\label{fig:2d_block_cyclic}
\end{figure}

Some scalable eigensolvers in ScaLAPACK, for example, PDSYGVX and PDSYEVX \cite{Scalapack},  are called to compute the eigenvalues for the distributed matrix.

At first, we would like to interpret the relationship between our numerical results and the true spectral gaps.
Due to the domain truncation and DG approximation, the numerical results may not represent the true spectral gaps; however, the convergence
Theorem \ref{thm:Conv_RayleighQuot} for the approximate Raleigh quotients in Section \ref{sec:Conv_Rayleigh} tells us that, if there exists a spectral gap for the true problem, then as long as
the domain is truncated large enough, what matters will be only the DG scheme approximation accuracy. And if there is no spectral gap, then as computing
domain gets larger, the numerical ``spectral gap" will clearly decay down to zero. This is exactly what Figure \ref{fig:sg2d_maxwell} and Figure \ref{fig:sg2d_soft}
are showing.

\begin{figure}[!htb]
\begin{minipage}[t]{0.45\linewidth}
\centering
\includegraphics[width=70mm]{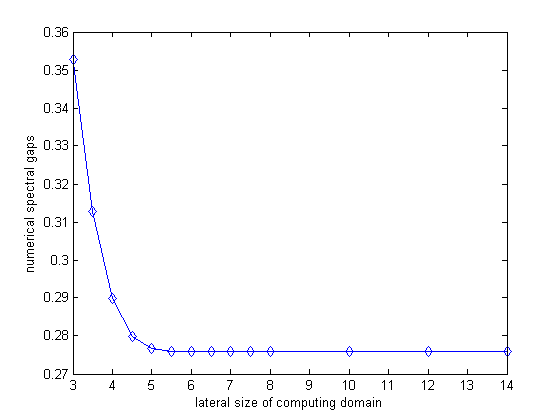}
\caption{The numerical spectral gaps for 2d Maxwell type model, i.e. $\gamma=0, \alpha=-1$}\label{fig:sg2d_maxwell}
\end{minipage}
\hfill
\begin{minipage}[t]{0.45\linewidth}
\centering
\includegraphics[width=70mm]{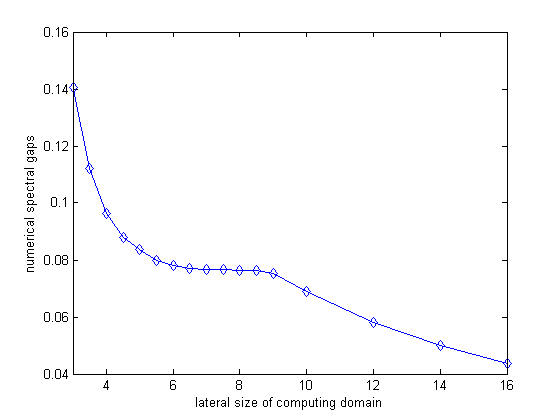}
\caption{The numerical spectral gaps for 2d, $\gamma=-1, \alpha=-1$}\label{fig:sg2d_soft}
\end{minipage}
\end{figure}

\emph{Note: }When increasing the lateral size of the truncated velocity domain, we keep the mesh size to be consistent (say, in our tests, $\Delta v$=0.5), for sake of comparison. For the case of soft potential,
as shown in Figure \ref{fig:sg2d_soft} for $\gamma=-1$, some ``pseudo spectral gap'' in the numerical results might be observed, for example in the segment $V\in[7,9]$; but such ``pseudo spectral gap'' will break immediately when increasing the domain size.

Then, we fix a large enough lateral size, increasing the number of mesh elements on each direction. More accurate results can be expected. We can see from Figure \ref{fig:sg2d_V5maxwell} and Figure \ref{fig:sg3d_V5maxwell},
the numerical values will approach the analytical value $\frac{1}{4}$ (for 2d) and $\frac{1}{3}$ (for 3d) respectively, when finer discretization is applied, as calculating the spectral gap for Maxwell type of interactions ($\gamma=0$),
where the exact eigenvalue for Maxwell-type interactions ($\gamma=0$) is known and given by \cite{ChangUhlenbeck,Cercignani_mathKinetic,Bobylev_linBE}:
\begin{equation*}
\lambda_{nl}=\int_{S^{d-1}}b(\cos(\theta))\left[\cos^{2n+1}\frac{\theta}{2}P_{l}(\cos(\frac{\theta}{2}))+\sin^{2n+1}\frac{\theta}{2}P_{l}(\sin\frac{\theta}{2})-1-\delta_{l0}\delta_{n0}\right] \, ,
\end{equation*}
where $P_{l}(x)$ is the $l$-th Legendre polynomial; $n$, $l$=0,1,.... Please note, when $d=2$, the above analytical calculation for spectral gap is only valid for constant angular cross-section $b$. 

\begin{figure}[!htb]
\begin{minipage}[t]{0.45\linewidth}
\centering
\includegraphics[width=70mm]{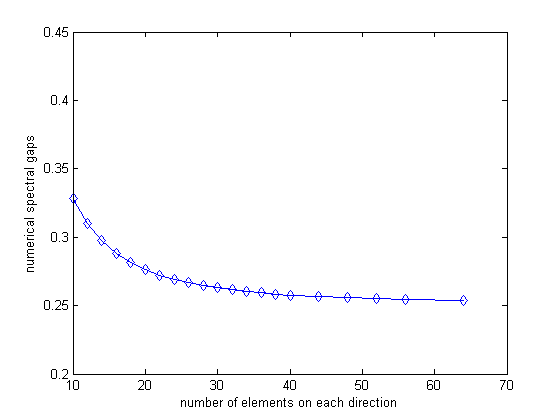}
\caption{The numerical spectral gaps with $\Omega_{v}=[-5,5]^{2}$ for 2d Maxwell type model, i.e. $\gamma=0, \alpha=-1$}\label{fig:sg2d_V5maxwell}
\end{minipage}
\hfill
\begin{minipage}[t]{0.45\linewidth}
\centering
\includegraphics[width=70mm]{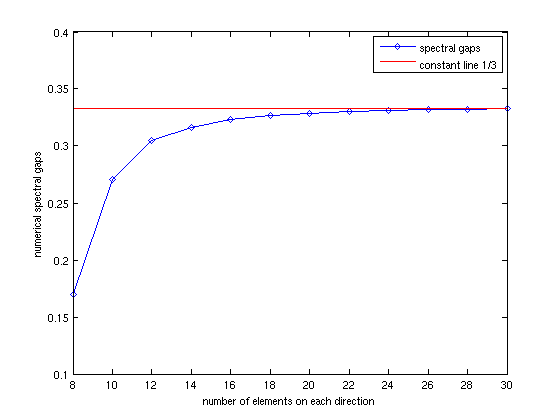}
\caption{The numerical spectral gap with $\Omega_{v}=[-5,5]^{3}$ for 3d Maxwell type model, $\gamma=0, \alpha=-2$}\label{fig:sg3d_V5maxwell}
\end{minipage}
\end{figure}

In particular, by actually solving the nonlinear Boltzmann equation and plotting the evolution of the weighted $L^{2}$ norm of the solution,
we can expect an exponential decay rate governed by or close to the spectral gap.
With the same DG discretization, the numerical value of the corresponding spectral gap for $\gamma=1$ (hars sphere) is 0.72.
The numerical solutions for the corresponding nonlinear BE is obtained by conservative DG solver developed also by the authors, see Chapter 3. See Figure \ref{fig:sg_exp_decay}.
\begin{figure}[!htb]
\centering
  \includegraphics[width=80mm]{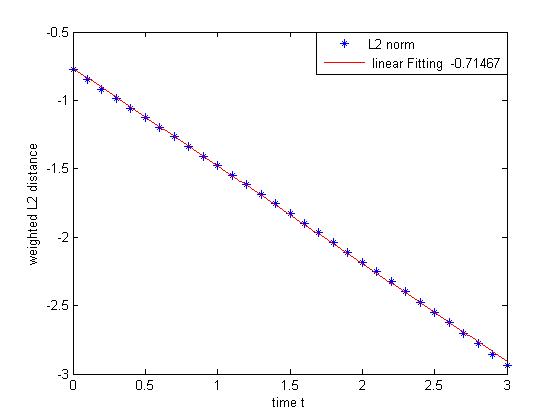}\\
  \caption{The exponential decay for solutions of 2d nonlinear Boltzmann equation with $\gamma=1, \alpha=-1$}\label{fig:sg_exp_decay}
\end{figure}

\emph{Remark. }This can only be expected after long time or with an initial state very close to equilibrium, because the spectral gap, as the first non-zero eigenvalue, can only dominate the decay rate when time $t$ is large enough.

We have computed spectral gaps for 2d variable hard potentials with isotropic angular cross-sections, using a moderate domain discretization (piecewise constant basis functions; $V=5$, $N=24$) . As seen from Table \ref{table:sg2d_vhp}, stronger intermolecular potential will force a faster decay to equilibrium.
\begin{table}[!ht]
\centering
\begin{tabular}{|c| c c c c c c c|}
  \hline
   $\gamma$ & 0 & 0.1 & 0.25 &0.5 & 0.75 & 0.9 & 1 \\
   \hline
  gaps & 0.25 & 0.29 & 0.34 &0.44 & 0.58 &0.67 &0.72 \\
  \hline
\end{tabular}
\caption{Numerical spectral gaps for 2d variable hard potentials with isotropic angular cross-sections}
\label{table:sg2d_vhp}
\end{table}

We also apply piecewise linear basis functions ($P^{1}$ polynomials) for approximating $F(v)$. Table \ref{Comp_sg_p0p1} is the comparison with piecewise constant case.
\begin{table}[!ht]
\centering
\begin{tabular}{|c| c c|}
  \hline
   gap & (V,N)=(5,20) & (V,N)=(5,24) \\
   \hline
  $P^{0}$ & 0.383798 & 0.353494 \\
  $P^{1}$ & 0.351826 & 0.332835 \\

  \hline
\end{tabular}
\caption{Comparisons of numerical spectral gaps between $P^{0}$ and $P^{1}$ basis, for 3d Maxwell model.}
\label{Comp_sg_p0p1}
\end{table}
from which one can easily see the $P^{1}$ basis functions give a much more accurate approximation than $P^{0}$, which is stated in the theorem of convergence.

For the non-cutoff cases, when $\int_{\mathbb{S}^{d-1}}b(\frac{u\cdot\sigma}{|u|})d\sigma$ is unbounded, we also have numerically verified the ``conjecture" on the existence of spectral gaps, i.e. there exists spectral gap if and only if $\gamma+\alpha\geq 0$. Numerical evidence shows that, similar to the cutoff case, the geometry of the spectral gaps for truncated operator $\chi_{\Omega}L$ also depends on the truncation of the domain and the discretization resolution. 
If there exists a spectral gap, as long as the computing velocity domain is large enough, the approximation accuracy only depends on the resolution of the mesh and vice versa;
 otherwise, if there is no spectral gap, with the lateral size getting larger, the numerical spectral gap is expected to decay to zero, and vice versa. See Figure \ref{fig:sg3d_VN_gamma0alpha0} and \ref{fig:sg3d_VN_gamma-1alpha0}. This is an interesting observation, to which we would like to provide some analytical explanations.

\begin{figure}[!htb]
\begin{minipage}[t]{0.45\linewidth}
\centering
\includegraphics[width=70mm]{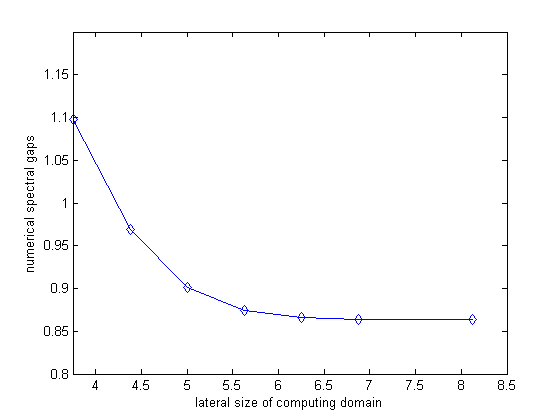}
\caption{The numerical spectral gaps for 3d non-cutoff case, $\gamma=0, \alpha=0$}\label{fig:sg3d_VN_gamma0alpha0}
\end{minipage}
\hfill
\begin{minipage}[t]{0.45\linewidth}
\centering
\includegraphics[width=70mm]{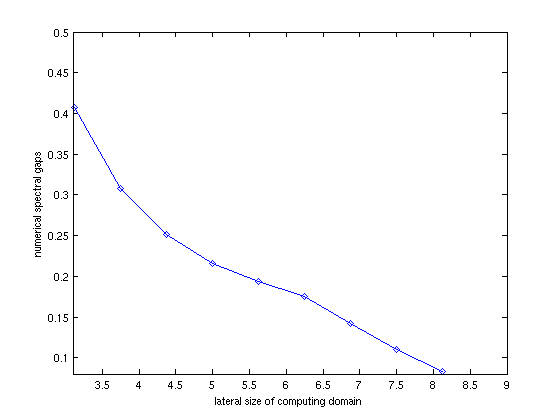}
\caption{The numerical spectral gaps for 3d non-cutoff case, $\gamma=-1, \alpha=0$}\label{fig:sg3d_VN_gamma-1alpha0}
\end{minipage}
\end{figure}

Therefore, once we know there exists a spectral gap, we can fix a large enough truncated velocity domain and apply DG meshes with finer resolutions, then more accurate approximations to the real spectral gap can be expected.
See Figure \ref{fig:sg3d_V5_gamma0alpha0} for the numerical spectral gaps when $\gamma=0, \alpha=0$, where an approximate value 1.0 is achieved.

\begin{figure}[!htb]
\centering
\includegraphics[width=70mm]{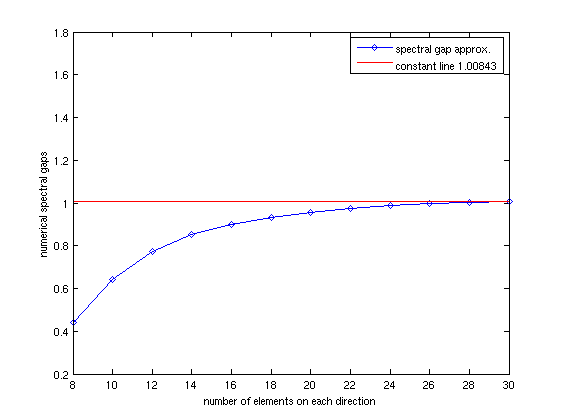}\\
\caption{The numerical spectral gaps with $\Omega_{v}=[-5,5]^{3}$ for 3d $\gamma=0, \alpha=0$}\label{fig:sg3d_V5_gamma0alpha0}
\end{figure}

We list the results for 3d variable hard potentials with isotropic angular cross-sections, see Table \ref{table:sg3d_vhp}.
\begin{table}[!ht]
\centering
\begin{tabular}{|c| c c c c c|}
  \hline
   $\gamma$ & 0 & 0.25 &0.5 & 0.75 & 1 \\
   \hline
  gaps & 0.33& 0.45& 0.62 &0.83 &1.10 \\
  \hline
\end{tabular}
\caption{Numerical spectral gaps for 3d variable hard potentials with isotropic angular cross-sections}
\label{table:sg3d_vhp}
\end{table}
from which we also can tell, as in 2d case, stronger intramolecular potential imposes faster decay to equilibrium.

\section{Summary}\label{sec:sum}

The existence as well as the quantitative information on the spectral gaps are very important for the justification of the Boltzmann model and study on the relaxation to equilibrium. This work is the first numerical verification, not only answering the existence of spectral gaps, but also provide numerical approximations to the real spectral gaps, if exist.

In this work, we have studied the geometry of spectral gaps for the linearized Boltzmann operators. For the integrable angular cross-sections, the Grad's splitting is recovered and used to build special approximation formulations. The Dirichlet form for the linearized operator is projected onto a Discontinuous Galerkin mesh, which results in an approximate Rayleigh quotient and can be proved to converge to the real spectral gaps. During the DG formulation, especially for the non-integrable angular cross-sections, a rotation transform has been applied to cancel the singularity in the angular cross-sections. The intersecting between $d-1$ dimensional sphere and the underlying DG mesh grids is also carefully analyzed, to guarantee accurate angular integrals over the sphere. The conservation routine is also applied to make a correction to the ``collision matrix". The conservation correction, equivalently, rules out the null space of the linearized Boltzmann operator.

A hybrid OpenMP and MPI paralleling computing is implemented to compute the eigenvalues of the conservative corrected ``collision matrix". Some routines in package like Scalapck \cite{Scalapack} have been called. Our test computations have been distributed among up to 256 nodes and 4000 cores on clusters Lonestar and Stampede affiliated with TACC \cite{TACC}. As long as memory and computing power allows, one can improve the accuracy of the numerical spectral gaps by choosing larger velocity domain, finer DG meshes and higher accuracy quadrature rules. This is also what we hope to do in future. With more efficient and accurate computing, one can explore more on the ``conjecture" and have a clear picture of the geometry of spectral gaps for different $\gamma$ and $\alpha$. And also, by considering the limit of Coulombic interactions ($\gamma=-3$, $\alpha=2$), one can answer the ``conjecture" \cite{MouhotStrain_SG} on spectral gaps for the linearized Landau operators.

\section*{Acknowledgement}

The authors thank Robert M. Strain for very valuable discussions that motivated this work back in 2013. I. Gamba and C. Zhang have been partially supported by NSF under grants DMS-1413064, DMS-1217154, NSF-RNMS 1107465 and the Moncreif Foundation. Support from the Institute of Computational Engineering and Sciences (ICES) at the University of Texas Austin is gratefully acknowledged.

\end{document}